\documentclass[11pt]{amsart}       
\usepackage{graphicx}
\usepackage{amsmath}
\usepackage{amsthm}
\usepackage{amsfonts}  
\usepackage{amssymb}
\usepackage{mathtools}
\usepackage{amscd}
\usepackage{tikz}
	\usetikzlibrary{cd}
\usepackage[bookmarks]{hyperref}
\DeclareMathOperator{\tr}{tr}		
\DeclareMathOperator{\ad}{ad}		
\DeclareMathOperator{\id}{id}		
\newcommand{\bndr}[1]{\partial{#1}}	
\newcommand{\isom}{\ensuremath{\cong}} 
\newcommand{\ext}{\ensuremath{d}} 
\newcommand{\cob}{\ensuremath{\delta}} 
\newcommand{\inv}[1]{\ensuremath{{#1}^{-1}}} 
\newcommand{\gext}[1]{\ensuremath{\widehat{#1}}} 
\newcommand{\extseq}[3]{\ensuremath{#1 \hookrightarrow \gext{#2} \to #3}} 
\newcommand{\lie}[1]{\ensuremath{\mathfrak{#1}}} 
\newcommand{\aut}[1]{\ensuremath{\mathrm{aut}(#1)}} 

\newcommand{\A}{\mathcal{A}} 
\newcommand{\smooth}[1]{\ensuremath{\textrm{Map}(#1)}}
\newcommand{\ballgroup}{\ensuremath{B^3_{\flat}G}} 
\newcommand{\basedball}{\ensuremath{B^3_{\flat,c}G}} 
\newcommand{\loopgroup}{\ensuremath{\Omega^3G}} 
\newcommand{\loopgroupmod}{\ensuremath{\mathcal{S}}} 
\newcommand{\spheregroup}{\ensuremath{S^3G}} 
\newcommand{\potgroup}{\ensuremath{\smooth{\ballgroup,S^1}}} 

\newtheorem{definition}{Definition}
\newtheorem{theorem}{Theorem}

\begin{document}

\title[A crossed module representation of  a 3-loop group]{A crossed module representation of a $2$-group constructed from the  $3$-loop group $\loopgroup$}

\author{Jouko Mickelsson}

\address[Jouko Mickelsson]{Department of Mathematics and Statistics, University of Helsinki}


 \email{jouko.mickelsson@gmail.com} 



\maketitle

\begin{abstract}
The quantization of chiral fermions on a 3-manifold in an external gauge potential is known to lead to an abelian extension of the gauge group. In this article we concentrate on the case of $\Omega^3 G$ of based smooth maps
on a 3-sphere taking values in a compact Lie group $G.$ There is a crossed module constructed from an abelian extension $\widehat{\Omega^3 G}$ of this group and a group of automorphims acting on it as eplained in a 
recent article by Mickelsson and Niemim\"aki. We shall
construct a representation of this crossed module in terms of a repesentation of $\widehat{\Omega^3 G}$ on a space of functions of gauge potentials with values in a fermionic Fock space and a representation of the automorphism
group of $\widehat{\Omega^3 G}$ as outer automorphisms of the canonical anticommutation relations algebra in the Fock space.

\end{abstract}

\section{Introduction}
\label{intro}

A crossed module of a pair of groups can be viewed as a higher categorical level of an ordinary group, called a 2-group. Motivated by the desire to understand the geometry of string groups as 2-groups the case of the
loop group $LG$ of a compact Lie group $G$ and the group of based paths $PG$ was studied in  \cite{bscs2007}. 

 In a recent work Kristel, Ludewig and Waldorf \cite{KLW1}, \cite{KLW2}  constructed  a representation of the 2-group (or equivalently, a crossed module) defined using the smooth loop group $L
G$  and the corresponding path
group $PG$ of smooth paths starting from the identity. The representaion was defined in a representation space (fermionic Fock space) of the canonical anticommutation relations algebra (CAR) based on a Hilbert space $H.$
A central extension $\widehat{LG}$ of the loop group acts unitarily in the Fock space whereas the elements of $PG$ act as outer automorphism of both  $\widehat{LG}$ and the CAR algebra, with the compatibility
conditions coming  from the definition of a crossed module. 

The aim of the present work is to extend the construction in \cite{KLW1} to the case of the group $\Omega^3 G$ of based smooth maps from  the 3-sphere to $G$ and an  analog of the path group.  The motivation
for the choice of the 3-sphere comes from 1) the fact that the nontrivial gauge group extensions exist in odd dimensions: This is related to the appearance of hamiltonian anomalies in odd dimensions, coming from characteristic
classes as functions of gauge potentials, using the descent equations starting from a Chern class in dimension $2n+2$ down to a Chern-Simons form in $2n+1$ dimensions and further to the gauge anomaly of the Dirac determinant
for a vector bundle over an $2n$ dimenional manifold, and finally to the commutator anomaly in $2n-1$ dimensions.  2) A physical motivation for choosing $S^3$  is that it can be considered as the compactification 
of the the space $\mathbb{R}^3$  in quantum field theory models. However, there should not be any real problems choosing another compact oriented 3-manifold. The generalization to higher odd dimensional
manifolds is also possible but the formulas for the gauge group extensions and their repsentations get more complicated.

 The first problem
comes with the fact that the group $\Omega^3 G$ is not (projectivly) represented in a Fock space. However, one can circumvent this obstruction by noting that the group is actually represented in a space of
functions of gauge connections (vector potentials) on $S^3$ with values in a Fock space. The projective phases are functions of the vector potential meaning that we do not have just a central extension but an
abelian extension $\widehat{\Omega^3G}$ by the group of circle valued functions of the vector potential.  The actual construction is based on the observation that one can fix a family $T_A$ of unitary operators in the Hilbert 
$H$ of square integrable spinor fields on $S^3$ as a function
of the external vector potential $A$  such that the operators $T_{A^g}^{-1} g T_A$  satisfy the condition  in \cite{ShSt} to be represented as operators in the Fock space; here $g$ denotes a multiplication operator
in $H$ corresponding to an element $g\in \Omega^3G$  and $A\mapsto  A^g$ is the gauge action to a vector potential $A,$ \cite{Mi94} and in a more general setting \cite{LaMi}. 

The analog of the path group is the group $\ballgroup$  consisting of smooth maps $g: B^3 \to G$ with appropriate boundary conditions on the boundary $S^2$ of the 3-ball $B^3.$ This group will also play 
another role in the discussion: The elements $g$ correspond to gauge potentials $A= g^{-1} dg$ on the 3-sphere $S^3$; because of the boundary conditions on $g$ we can identify the boundary as one
point and so the domain of $A$ is actually $S^3$ with the boundary of $B^3$ contracted to a point.

For the crossed module we need the group $\ballgroup$ (defined in Section 2) which corresponds to the path group $PG$ and contains as a subgroup the group $\Omega^3 G.$  A complication compared with the
loop group case is that $\ballgroup$ does not act as automorphims on the extension $\widehat{\Omega^3G}$ but we have to extend the former to a group $\gext{\ballgroup}$ by an abelian normal subgroup.
So  we  are looking for a representation of the crossed module $(\gext{\ballgroup}, \widehat{\Omega^3G}).$ Similarly as in the loop group case the group $\gext{\ballgroup}$ acts as outer automorphisms
of the CAR algebra. In order that we can call this a representation of the crossed module we need to show that the outer automorphism group is compatible with the action of the inner automorphisms 
coming from the Fock space representation of $\widehat{\Omega^3G}.$

All the groups of functions taking values in a Lie group are considered  as infinite-dimensional Fr\'{e}chet Lie groups, see \cite{Neeb1} \cite{KM}. This is because we are using the results in \cite{MN} on the construction of the crossed module
(recalled as  Theorem 1 in the present paper)  which in turn use the work by Neeb on infinite dimensional Lie groups, \cite{Neeb1}.

One of the strengths of the original  loop group construction is that it allows the action of the circle group. In three dimensions one would like to have the corresponding action of $SO(4)$, but there seems to be no straight-forward way to incorporate this into the picture presented here. The extension of the $3$-loop group necessitates a fixed point in $B^3$ -- namely, the contracted boundary $S^2$ -- and this choice cannot be equivariant under the symmetry action but only under
the base point preserving rotation group $SO(3)$ on $S^3.$

In Sections 2 and 3 we recall the construction of the crossed module  $(\gext{\ballgroup}, \widehat{\Omega^3G})$ from \cite{MN} and finally in Section 4 we construct a CAR algebra representation of this crossed module.

\paragraph{Notation and conventions}

Unless otherwise stated, $G$ is a simply-connected Lie group.  The group identity element is denoted by $e$ throughout, and we write $PG$ for the group of based paths in $G$. The Lie algebra of $G$ is denoted by $\lie{g}$. We identify the $1$-sphere $S^1$ with the circle group.

In places there is an implicit assumption to consider the connected component of a given group in case the connectedness for the whole group is not available. The reason is that the construction of extensions of the
groups of gauge transformations leads to extra complications in the non connected case which we want to avoid; as an example, see Section 2.2 in \cite{MN}. 

\section{$3$-loop group and its abelian extension}
\label{sec:2}
In this Section we recollect some of the results in \cite{MN}.
Let $\spheregroup$ be the group of smooth maps from the $3$-sphere to a Lie group $G$. Since every $n$-sphere can be given as the quotient $B^n/S^{n-1}$, where the boundary of the $n$-ball is contracted to a point (call it the north pole of the sphere $S^n$), we define by analogy the flattened group 
\[
	\ballgroup = \{ f \in \smooth{B^3,G} : \partial_r^{(n)} f = 0 \textrm{ for $n=1,2,3, \dots$  on the boundary } S^2 \} ,
\]
where the group multiplication is as usual the point-wise multiplication in the domain and $\partial_r$ is the partial  derivative in the radial direction at the boundary.
 We then have the \emph{$3$-loop group} defined as
\[
	\loopgroup = \{ f \in \ballgroup : f \textrm{ extends to } S^3 \textrm{ and } f(S^2) = e \} .
\]
The extension above is defined by thinking of $S^3$ as the closed unit ball with the boundary contracted to one point.
Note that while the group of rotations $SO(4)$ acts on the sphere group $S^3G$, it does not act on $\loopgroup$ since the boundary of $B^3$ is the contracted fixed point and the action cannot be extended to $S^3$.

The standard Fr{\'e}chet topology on groups and Lie algebras of smooth maps on a compact manifold is defined by the infinite set of seminorms coming from the supremums of derivatives of all orders. In this topology
the subgroups below are closed and the quotients are again in the same category of groups. For more information on infinite-dimensional locally convex Lie groups, which can be applied to the groups at hand, see
\cite{KM}, \cite{Neeb2}

It is easy to see that $\ballgroup$ retains the Lie group structure of $\smooth{B^3,G}$, and likewise for $\loopgroup$. In the case of one-dimensional loop group $\Omega G$ and the free loop group $LG$, there is the split exact sequence
\[
	\Omega G \hookrightarrow LG \to G .
\]
For the $3$-loop group this relationship is retained as follows. If we denote by $\basedball$ the subgroup of $\ballgroup$ of maps that are constant on the boundary $S^2$, the group $\basedball$ is in fact a principal $\loopgroup$-bundle over $G$:
\[
	\loopgroup \hookrightarrow \basedball \xrightarrow{\phi} G ,
\]
where $\inv{\phi}(e) = \loopgroup$ is the canonical fibre. 

Furthermore, we note that $\loopgroup$ is a split normal subgroup of $\ballgroup$. The quotient map
\[
	q: \ballgroup \to \ballgroup/\loopgroup \isom \smooth{S^2,G}
\]
has $\loopgroup$ as its kernel.

Given any mapping group $\smooth{X,G}$ for a Lie group $G$ and a manifold $X$, there is an extension by the Abelian group of smooth maps $\smooth{\smooth{X,G},S^1}$~\cite{ps1986}, pp. 66 - 67. Let us then consider another  Abelian extension \cite{Mic87}    
\[
	\extseq{\potgroup}{\spheregroup}{\spheregroup} .
\]
Here the group $\spheregroup$ acts in the fiber through point-wise  right multiplication on functions, $f\mapsto fg$ for $g\in \spheregroup.$  The fiber contracts to $S^1$ (the constant maps taking values in $S^1$) and 
the topological nontriviality of the fibering comes from the nontriviality of circle bundle over $\spheregroup.$
In the gauge-theoretical formulation  \cite{Mic87} the fibre was actually the group $\smooth{\A, S^1}$ where $\A$ is  the space of vector potentials $A$ on  $S^3.$ However, using the map $f\mapsto f^{-1}df= A$ one can associate 
to $f\in \ballgroup$ a gauge potential  such that $A$ is an exact gauge on the boundary and the construction in \cite{Mic87} works as well in this setting.
 $\ballgroup$  is actually a good model for the gauge potentials on $S^3$ since modulo gauge
transformations $g$ such that $g(S^2) =  e$ the moduli space of these potentials with values in the Lie algebra $\lie{g}$ is equal to maps $S^2 \to G$  which is equal to the moduli space of gauge potentials (modulo based gauge transformations) on $S^3$  \cite{Sin81}. See for the motivation also in  \cite{mw2016}.

On the Lie algebra level we have the corresponding Mickelsson-Faddeev cocycle
\begin{equation}\label{eq:s3g-cocycle}
	\theta(A;x,y) =   k \int_{S^3} \tr A [\ext x, \ext y] .
\end{equation}
where $x,y$ are smooth maps from $S^3$ to the Lie algebra  $\lie{g}$ and
the trace is evaluated in a representation of $\mathfrak{g}$ and the normalization $k$ depends  on the representation. For example, in the defining represention of $SU(n)$ the basic cocycle
corresponds to $k = \frac{\sqrt{-1}}{12\pi^2}.$ The vector potential $A$ is a Lie algebra valued 1-form on $S^3$ and so the integrand is indeed  a 3-form on $S^3.$
No modification is needed to write this in terms of the $3$-loop group:
\[
	\extseq{\potgroup}{\loopgroup}{\loopgroup} .
\]
The abelian extension $\widehat{S^3G}$ gives by restriction of the base to the subgroup $\Omega^3 G \subset S^3G$ an abelian extension $\widehat{\Omega^3 G}$ with the same fiber.
The corresponding Lie algebra cocycle is exactly the same as above but the 1-form $A$ is   a function of $f\in \ballgroup$  using the map $f\mapsto f^{-1}df =A.$

\section{The  crossed module produced from the extensions of $\Omega^3G$ and $\ballgroup$}

\begin{definition}[Crossed module]
Let $G$ and $H$ be groups, and consider morphisms
\[ 
	\delta : H \to G \quad \textrm{and} \quad \alpha : G \to \aut{H} .
\]
We say that  these morphisms define   a crossed module if the following two diagrams commute.
\[
\begin{tikzcd}
H \times H \arrow[rd,"\ad"] \arrow[r,"\delta\times\id"] & G \times H \arrow[d,"\alpha"]\\
	& H
\end{tikzcd}
\qquad
\begin{tikzcd}
G \times H \arrow[d,"\id\times\delta"] \arrow[r,"\alpha"] & H \arrow[d,"\delta"] \\
G \times G \arrow[r,"\ad"] & G
\end{tikzcd}
\]
Equivalently, if we denote by $h^g$ the element-wise action of $G$ on $H$, the diagrams correspond to the equations
\[
	h^{\delta(h')} = \inv{h'} h h'
\]
and
\[
	\delta(h^g) = \inv{g} \delta(h) g 
\]	
for all $h,h' \in H$ and $g \in G$.
\end{definition}

The first equation is called the Peiffer identity. The second is a statement about the equivariantness of the $G$ action on $H.$ 
The definition goes back to the article in 1946 by Whitehead \cite{W}. And its specialization to the smooth setting is straight forward:

\begin{definition}[Smooth crossed module]
If the groups $G$ and $H$ in a crossed module  are Lie groups and the action defined by the morphism $\alpha$ is smooth, the crossed module is called a Lie crossed module, or a smooth crossed module.
\end{definition}

 Let us extend the group $\ballgroup$ as follows.  First we define a Lie algebra extension of $B^3_{\flat} \lie{g}$ by the abelian ideal  $\textrm{Map}(B^3_{\flat}G, i\mathbb{R})/i\mathbb{R}$  by using the point-wise action of
the former to the argument of the functions in the latter and twisting the semi-direct sum by the 2-cocycle in (1) but now allowing $x,y$ to be functions on $B^3_{\flat}.$ Taking the latter space of smooth functions
modulo the constants $i\mathbb{R}$ is necessary to avoid the constant term in the coboundary of (1)  (where $A= f^{-1} df$ for $f\in \potgroup)$,

\[
	\cob \theta = k \int_{\bndr B^3} \tr(x [\ext y, \ext z] - y [\ext z, \ext x] + z [\ext x, \ext y]) .
\]

The Lie algebra extension determines  uniquely  the structure of the contractible (and in particular simply connected)  covering group of the Lie algebra extension by Thm. VII.2  \cite{Neeb04} (cited in Thm. 2.7 in \cite{MN}) ,


$${\potgroup/S^1} \to  {\potgroup/S^1} \rtimes_{\theta} {\ballgroup}   \to {\ballgroup}$$
where all the groups are contractible and thus the fibering is topologically trivial. We skip the explicit  construction of the Lie group extension but it should be rather straight forward using  the method in
\cite{Mic87} for the construction of  the abelian extension of $S^3 G.$

Note that now the right adjoint action of $\gext{\ballgroup}$ defines automorphisms on the Lie algebra of the extension $\gext{\loopgroup}$ so that we have a homomorphism
\[
	\gext{\ballgroup} \to \aut{\gext{\Omega^3\lie{g}}} ,
\]
As shown in \cite{MN}:
If the (connected component of the identity of the) group $\gext{\loopgroup}$ is simply connected, then the homomorphism
\[
	\alpha' : \gext{\ballgroup} \to \aut{\gext{\Omega^3\lie{g}}}
\]
lifts to a homomorphism
\[
	\alpha : \gext{\ballgroup} \to \aut{\gext{\loopgroup}} 
\]
such that the action  on $\gext{\loopgroup}$ is smooth.

On the other hand, whether $\gext{\loopgroup}$ is simply connected or not, we can consider the group
\[
	\loopgroupmod = \gext{\loopgroup}/ S^1 = \potgroup/S^1 \rtimes \loopgroup
\]
which is a normal subgroup of $\gext{\ballgroup}$ since $\loopgroup$ is normal in
$\ballgroup$ and the fiber $\potgroup/S^1$ is mapped  onto itself in conjugation.  Hence there is a smooth right adjoint action of $\gext{\ballgroup}$ on $\loopgroupmod$ by conjugation. Furthermore, while the previously introduced $3$-loop group extension is not central, there is a central extension
\begin{equation}\label{eq:central_ext}
	\extseq{S^1}{\loopgroupmod}{\loopgroupmod} ,
\end{equation}
where the group $\gext{\ballgroup}$ acts trivially on the central $S^1$. Note that this is none other than the original extension $\gext{\loopgroup}$ of the $3$-loop group.

Furthermore,  we have a  lifting of the homomorphism \cite{MN}
\[
	\psi: \gext{\ballgroup} \to \aut{S^1} \times \aut{\potgroup/S^1 \rtimes \loopgroup}
\]
to a homomorphism
\[
	\gext{\psi} : \gext{\ballgroup} \to \aut{\gext{\loopgroup},\potgroup}
\]
.

The content of the Theorem below is just a repetition of   the last three equations in Section 3 in \cite{MN}, together with the definition of a crossed module.

 \begin{theorem} 
 There is a smooth homomorphism
\[
	\alpha: \gext{\ballgroup} \to \aut{\gext{\loopgroup}} ,
\]
 
which can be in fact extended to a crossed module using the homomorphism $\gext{\loopgroup} \to \gext{\ballgroup}$  which is the composite map

of the  epimorphism
\[
	\gext{\loopgroup} \to \potgroup/S^1 \rtimes \loopgroup ,
\]
to the normal subgroup $\potgroup/S^1 \rtimes \loopgroup $ of $\gext{\ballgroup}$  and the natural inclusion
\[
	\potgroup/S^1 \rtimes \loopgroup  \hookrightarrow \gext{\ballgroup} .
\]

\end{theorem}

\section{CAR algebra representation}

In the case of the loop group $\Omega G$ there is a Fock space representaion of the crossed module $(PG, \widehat{\Omega G})$ where the group $PG$ of smooth paths in $G$ starting from the identity acts as automorphisms
of the standard central extension $\widehat{\Omega G}$ and the construction of the crossed module is completed by the composite map of the canonical projection $\widehat{\Omega G} \to \Omega G$ and the natural
embedding $\Omega G \to PG.$ The representation is defined through a representation of $\widehat{\Omega G}$  in a representaion of the canonical anticommutation relations algebra (CAR) in  a complex Hilbert space
$\mathcal{F},$  \cite{KLW1}.  The latter representation is fixed by a polarization $H= H_{+} \oplus H_-$ of a complex Hilbert space $H.$ To each vector $v\in H$ is associated a pair $a^*(v), a(v)$ in the CAR algebra, the first is linear
in $v$ whereas the second is antilinear, and they obey the anticommutation relations
$$  a^*(v) a(u) + a(u) a^*(v) = <u,v> \mathbf{1}$$
where $<\cdot, \cdot>$ is the inner product. All the 'creation operators' $a^*$ anticommute among themselves and the same with the 'annihilation operators' $a.$  The Hilbert space $\mathcal{F}$ is defined by the existence of
a vacuum vector $\psi$ such that $a^*(u) \psi = 0 = a(v) \psi$  for all $u\in H_-$ and $v\in H_+.$ In this representation $a^*(v)$ is the adjoint of $a(v).$ The norm of $\psi$ is set to be $1.$ 

In the case of the loop group representation $H$ is selected as the space of square integrable functions on the unit circle with values in the defining representation of the simple Lie group $G.$ The polarization is fixed
by setting $H_+$ equal to the subspace spanned by the nonnegative Fourier modes.

The automorphism group $PG$ of $\widehat{\Omega G}$ acts as automorphism of CAR algebra by $a^*(v)  \mapsto a^*(fv)$ where $f$ acts on $v$ by pointwise multiplication.  This defines also an inner automorphism
on the unitaries representing the loop group in the Fock space as follows:  Let $\hat{g} \in \widehat{\Omega G}$ and $g$ its projection in $\Omega G.$ Then  $\hat g a^*(v) \hat g^{-1} = a^*(gv).$
We have denoted an element of the loop group and its representative in the Fock space by the same symbol. Now we can compose the action of an inner automorphism $\hat g$ and an outer automorphism $Aut_f$
for $f\in PG$ on the CAR algebra as
$$Aut_f Aut_{\hat g} Aut_f^{-1} a^*(v)  =  a^*(f g f^{-1} v).$$
This is an inner automorphism by the element $h= fgf^{-1}$ of the loop group and therefore the action above can be written as
$$a^*(v)  \mapsto \hat{h} a^*(v) \hat{h}^{-1}$$
for an element $\hat h$ inthe central extension projecting onto $h.$ It follows from the irreducibility of the CAR representation  that the operator in the Fock space representing the inner automorphism $h$ is
uniquely defined up to a multiplicative phase and so indeed $Aut_f$ lifts to an automorphism of the group of unitary operators representing $ \widehat{\Omega G}$  in $\mathcal{F}.$  The phase can be fixed
by selecting $\hat h$ to be the representative of $Aut_f(\hat g) \in \widehat{\Omega G}$ in the Fock space. In this way the $PG$ action on $\widehat{\Omega G}$ is compatible with the outer automorphism action of
$PG$ on the CAR algebra and the representation of $\Omega G$ as inner automorphisms of the same CAR algebra.

Our next aim  is to extend the loop group construction above to  the case of the crossed module in the previous section. The first problem is that an element of $\Omega^3 G$  cannot be lifted to an operator in a Fock space.
This is due to the Shale-Stinespring theorem \cite{ShSt}  according to which the an operator  $g$ in $H$ is represented in $\mathcal F$ if and only is it satisfies the conditon that the off-diagonal blocks of $g$ in the polarization
$H= H_+ \oplus H_- $ are Hilbert-Schmidt. In our case the natural polarization is given by the spectral projections of a Dirac operator to negative and nonnegative eigenvalues. The Dirac operator at hand is acting
on spinor fields tensored with a representation of $G$ over $S^3.$ (In the case of the loop group the Dirac operator is simply $i \frac{d}{d \phi}$ on the unit circle.) 

 It is a simple calculation in pseudodifferential analysis involving the symbol of
the sign of the Dirac operator to see that the the condition is  satisfied if only if $g$ is a constant function on $S^3.$  In fact, in local cordinates $x_i$ and corresponding momenta $p_i$ the principal symbol of a Dirac operator is  $\sum \gamma^i  p_i$  and its sign $\epsilon = \sum \gamma^i p_i / |p| .$  Therefore the highest degree part of the commutator  $[\epsilon, g] = \sum \frac{\partial\epsilon}{\partial p_i} \frac{\partial g}{\partial x_i} + \text{terms  lower order in $p$} $ is of order
$-1$ in monenta in general, and zero  only if $g$ is a constant function.  But an operator symbol on a 3-manifold is Hilbert -Schmidt only if it is of order less than $- 3/2.$ 

However, the Hilbert-Schmidt condition can be circumvented if one works  in the completion of  the algebraic tensor product   $V\otimes H$,  where $V$ is the vector space of smooth complex valued functions on the infinite-dimensional
space $\ballgroup$ (with Fr{\'e}chet topology),   to the space consisting of all (Fr{\'e}chet) smooth functions on $\ballgroup$ with values in $H.$ The use of Fr{\'e}chet topology is reasonable in this context since the Lie algebra of the
gauge group is acting as differential operators  on $V.$ 

The method is based on the following fact, \cite{Mi94}, \cite{LaMi}. Given the family of Dirac operators $D_A$ parametrized by smooth vector potentials (Lie algebra valued 1-forms) on
$S^3$  one can construct a family of unitary operators $T_A$ acting in $H$ such  that
$$ \omega(g;A) = T^{-1}_{ A} g T_{A^g}$$
satisfies the Hilbert Schmidt condition. Here $ A^g = Aut_g(A) + g^{-1} dg$ is the gauge action on a vector potential $A.$  By construction $\omega$ is a 1-cocycle meaning that

\begin{equation} \label{eqn: 1cocycle}
\omega(g_1; A) \omega(g_2 ; A^{g_1}) = \omega(g_1 g_2; A).
\end{equation}

Actually, the same formula in \cite{Mi94} used to define $\omega(g; A)$ for $g\in \Omega^3 G$ can be extended to all $g\in \ballgroup$ but then the Hilbert-Schmidt condition is no more valid outside $\Omega^3G.$ 

In the present discussion we replace the space $\mathcal A$ of vector potentials by $\ballgroup$ using the embedding $\ballgroup \to \mathcal A$ by $f \mapsto f^{-1} df.$ Then the gauge transformation $A\mapsto A^g$ corresponds
to $f\mapsto fg$ for $g\in \Omega^3 G \subset \ballgroup.$ 

The action of $g\in \Omega^3 G$ on $V\otimes H$ is
$$v \otimes x \mapsto g\cdot v \otimes \omega(g; \cdot ) x$$
where $(g \cdot v)(f) = v(fg),$ pointiwise multiplication in the argument. By the cocycle property of $\omega$ we have then $(g_1 g_2) (v\otimes x) = g_1 (g_2 (v\otimes x)).$ 

Because $\omega$ satifies the Hilbert-Schmidt condition the action can be lifted to an  action on $V\otimes \mathcal F,$
\cite{ShSt} (again, we pass to the completion containing all smooth functions on $\ballgroup$ with values
in $\mathcal F.$) This action is given by the gauge action on the argument in $\ballgroup$ and a multiplication by an operator $\hat{\omega}(g;\cdot)$ in the Fock space. Here $\hat{\omega}$ is an element
of the central extension $\hat{U}_{res}$ of the restricted unitary group in $H$ consisting of unitaries with Hilbert-Schmidt operators in the off-diagonal blocks with respect to the given polarization.
The lift $\omega(g;A) \to \hat{\omega}(g;A)$ is again unique up the a phase. The important point now is that we do not get a central extension of $\Omega^3G$ since the phase depends on the  argument $A.$
Instead, we get an abelian extension by $Map(\ballgroup, S^1)$ with an action of $g\in \Omega^3 G$ by right multiplication $f\mapsto fg$ on the argument  $f\in \ballgroup.$ 

We  denote by the same symbol $\pi$ both the canonical projections $\widehat{\Omega^3 G} \to \Omega^3 G$ and $\widehat{\ballgroup} \to \ballgroup.$
Smooth functions on $\ballgroup$ with values in the CAR algebra are denoted by $\mathcal B.$ Each element $h\in\widehat{\ballgroup}$ defines an outer automorphism $\alpha_h$
on elements  $F\in \mathcal B$ by the action on the CAR algebra by $a^*(x) \mapsto a^*(\omega(\pi(h);A)x)$ and the gauge action on the argument $A= f^{-!}df$ of $F$
by $f\mapsto f\pi(h).$

\begin{definition}  A CAR representation of the crossed module $(\widehat{\ballgroup}, \widehat{\Omega^3 G})$ consists of 
1) representation of $\widehat{\Omega^3 G}$ as unitary operators in $\mathcal F$ as functions of $\ballgroup$  with gauge action $f\mapsto fg$ by $g\in\Omega^3 G$ on the
argument $f\in\ballgroup,$

2) outer automorphims of $f\in \widehat{\ballgroup}$ acting on elements of $\mathcal B$ composed of the right  action of $ \pi(f) \in \ballgroup$ on functions $F$ on $\ballgroup$ and of
outer automorphisms of the CAR algebra as $a^*(x) \mapsto a^*(\pi(f) x)$ where $\pi(f)$ is acting on $v\in H$ 

3)  compatibility condition: 
The combined conjugation by the automorphisms  $\alpha_h^{-1} \circ \beta \circ \alpha_h$ where $\alpha_h$ is the
outer automorphism by $h$ as defined above and $\beta_g$ is the (inner) automorphism on the CAR algebra through the  conjugation by $\hat{\omega}(g; \cdot)$ is the same as the action by the inner automorphism by $\gamma=\hat{\omega}(\pi(h)^{-1} g \pi(h);\cdot)$
and the right action on the argument on $F$ by $\pi(h)^{-1} g\pi(h).$ 
\end{definition}

{\bf Remark} The Peiffer identity of a crossed module is preserved in the CAR representation since the outer automorphisms on the CAR algebra by elements of $\widehat{\Omega G}$  are implementable 
(by the Hilbert-Schmidt condition) and thus are actually inner outomorphims.  The compatibility condition above corresponds to the second  equation in the definition of a crossed module.

 \begin{theorem}  We have a CAR representation of the crossed module $(\widehat{\ballgroup}, \widehat{\Omega^3 G})$ in the space $\mathcal B$  constructed from the family of inner automorphisms 
$\hat{\omega}(g; \cdot)$ and the outer automorphisms $\alpha_h$ for  $h\in \widehat{\ballgroup}.$ 

\end{theorem}

\begin{proof} 1)  We have already explained the construction of the operators $\hat{\omega}(g; \cdot)$ in the Fock space acting as inner outomorphims of the CAR algebra, 2) is clear from the definitions of its
ingredients. What needs a proof is the compatibility condition 3).

 First we observe that the action on the base of the extension $\widehat{\Omega^3 G} \to \Omega^3 G$ and the gauge action on the fiber $Map(\ballgroup,S^1)$ uniquely define
the conjugation action by  elements of $\widehat{\ballgroup}.$  Note that the conjugation action of any element of the abelian group $Map(\ballgroup, S^1)$  on the CAR algebra is trivial so that when looking for a group
element  in $\widehat{\Omega^3 G}$ which provides the required inner automorphism acting on the CAR algebra there is always an ambiguity modulo a conjugation acion by $Map(\ballgroup, S^1)$. 
What we want to show is that there is a natural  choice $\gamma.$ 
It is sufficient to check the conjugation action on the generators $a^*(x)$ of the CAR algebra: the algebra consists in polynomials of these elements and their adjoints $a(x).$
Now from the defintions of the incredients in $\gamma$ we get at a point $A$

\begin{align}
 (\alpha_h^{-1} \circ \beta_g \circ \alpha_h) a^*(x) 
= &\null (\alpha^{-1}_h \circ \beta_g) a^*(\omega(\pi(h); A)x)\\
= &\null \alpha^{-1}_h a^*(\omega(g; A) \omega(\pi(h); A^g)x)  \\
= &\null a^*(\omega(\pi(h); A) \omega(g; A^{h^{-1}}) \omega( \pi(h); A^{h^{-1}g} )x )\\
= &\null a^* (\omega(\pi(h^{-1}) g \pi(h); A) x)
\end{align}

where the last equation follows from cocycle property  \eqref{eqn: 1cocycle}  of $\omega.$  But the last expression is just the action of $\gamma$ on $a^*(x).$

\end{proof}

Conflict of interest statement: On behalf of all authors, the corresponding author states that there is no conflict of interest.

Data availability statement: No additional data available.

{}

\end{document}